%% file: Paper1-Method.tex
\documentclass[aps,prx,twocolumn,nofootinbib,10pt]{revtex4-2}

\usepackage[english]{babel}

\usepackage{amsmath}
\usepackage{mathtools}
\usepackage{amssymb}
\usepackage{amsthm}
\usepackage{amsmath}
\usepackage{bm}
\usepackage{comment}
\usepackage{graphicx}
\usepackage[ruled,linesnumbered]{algorithm2e}
\usepackage{algpseudocode}
\usepackage{braket}
\usepackage{tikz}
\usetikzlibrary{quantikz2}

\newtheorem{theorem}{Theorem}[section]

\newtheorem{lemma}[theorem]{Lemma}
\newtheorem{definition}[theorem]{Definition}
\usepackage[nameinlink,noabbrev]{cleveref} 

\begin{document}

\title{Constraint-oriented biased quantum search for linear constrained combinatorial optimization problems}
\author{Sören Wilkening}
\affiliation{Institut f\"ur Theoretische Physik, Leibniz Universit\"at Hannover, Germany}
\author{Maximilian Hess}
\affiliation{Infineon Technologies AG, Neubiberg, Germany}
\author{Timo Ziegler}
\affiliation{Institut f\"ur Theoretische Physik, Leibniz Universit\"at Hannover, Germany}

\begin{abstract}
    In this paper, we extend a previously presented Grover-based heuristic to tackle general combinatorial optimization problems with linear constraints.
    We further describe the introduced method as a framework that enables performance improvements through circuit optimization and machine learning techniques.
    Comparisons with state-of-the-art classical solvers further demonstrate the algorithm's potential to achieve a quantum advantage in terms of speed, given appropriate quantum hardware.
\end{abstract}

\maketitle

\section{Introduction}

Grover's search algorithm \cite{Grover1996AFQ} is one of the longest-standing achievements in the field of quantum computing. It enables finding a marked element in a list of $N$ members with $O(\sqrt{N})$ queries to an oracle that is capable of detecting the marked state. In contrast, any classical method has to do $O(N)$ queries to a classical version of this oracle.
Grover's algorithm works best for search tasks with no additional structure, i.e., when there are no systematic or efficient ways to find the desired element. Thus, its quadratic speedup is only achieved over brute force search.
While in the worst case, combinatorial optimization problems \cite{du1998handbook} are expected to be solvable only in exponential time, they do, in the average case, exhibit quite a lot of structure, which is used by algorithms such as Branch and Bound \cite{Lawler1966BnB} or Dynamic Programming \cite{Martello1999COMBO} to find optimal solutions efficiently for a large number of practical instances.
For a sensible application of Grover's algorithm to combinatorial optimization, one therefore needs to exploit the problem's structure rather than quadratically ``improving" a brute-force search.
A successful showcase of this idea is the 0-1-knapsack \cite{Pisinger1998KnapsackProblems} algorithm QTG proposed in \cite{wilkening2024quantumalgorithmsolving01}.
The core idea underlying the QTG is a problem-specific state preparation step, which
filters out infeasible states from the initial superposition and introduces a bias towards states close (in terms of Hamming distance) to a reference state.

More precisely, consider the 0-1 knapsack problem with $n$ items. Given a list of weights $w_1,...,w_n \in \mathbb{N}$ and values $v_1,...,v_n \in \mathbb{N}$ together with a weight capacity $C \in \mathbb{N}$, the QTG prepares a state

\begin{equation}
    \ket{QTG} \coloneqq \sum_{x \in \{0,1\}^n: w^T x \leq C} a_x \ket{x}
\end{equation}

with the amplitudes $a_x$ depending largely on the Hamming distance between $x$ and a reference state $x_0$. The main ingredient of the state preparation circuit is a series of (biased) Hadamard gates, each controlled on a feasibility condition verified in an auxiliary register which stores the ``constraint consumption".
This state preparation routine, together with amplitude amplification \cite{Brassard_2002}, yields quite positive results, as shown in \cite{wilkening2024quantumalgorithmsolving01}.
This work aims at generalizing the techniques used by the QTG algorithm to other constrained optimization problems.
First, we show how to handle problems with a single linear constraint whose coefficients can be both negative and positive. In that case, the techniques used for the QTG algorithm remain largely applicable, only the initial state (the ``root" of the quantum tree) has to be adjusted to the signs of the constraint coefficients.
Next, we describe how to handle problems with two or more linear constraints. For these problems, it is in general not expected that one can efficiently prepare a superposition of exactly the feasible states, as the preparation of a feasible state constitutes in itself an NP-hard \cite{karp2009reducibility} problem in some cases, e.g. the subset sum problem, see \cite{Horowitz1974}. However, it is possible to filter out a significant part of infeasible states using similar techniques, i.e., iteratively branching, based on one feasibility criterion for each constraint, from an initial state. This does not guarantee that we are able to filter out all infeasible states, as it is common that one reaches a ``dead end", i.e., a partial assignment of variables which can not be continued in a feasible way as all further assignments would lead to infeasibilities due to one of the constraints.
Still, we find in our numerical experiments that a significant amount of infeasible states is pruned. Together with a biasing mechanism based on neighborhood search (as in the QTG algorithm), this state preparation lays the groundwork for a successful Grover search as the results show.
In this sense, we call our framework Constraint-oriented biased quantum search (CBQS).
The problem under consideration in this paper is a variant of the knapsack problem with a minimum filling constraint in addition to the usual capacity constraint.
We obtain numerical results for our methods through a benchmarking process which does not involve the emulation of quantum computers, hence we are able to analyze instances beyond the toy size which experimental quantum algorithm works are usually confined to. We explain our benchmarking process and justify that the results obtained in this way are a meaningful representation of what actual quantum computers would produce.

\section{Previous Work}
Preparing quantum states that meet certain requirements, e.g., being optimal with respect to a functional while satisfying a range of constraints, is an endeavour almost as old as the field of quantum computing itself.
The famous Grover's algorithm \cite{Grover1996AFQ, ambainis2005quantumsearchalgorithms}, promising a quadratic speedup over unstructured search is a promising candidate for this purpose but lacks the consideration of problem structure if taken in its original form. 
Nested quantum search \cite{Cerf_2000}, on the other hand, mitigates this shortcoming to a certain degree but still does not provide relevant benefit over classical methods without a specific state preparation routine.
Benchmarking primitives to go beyond what is classically simulable were first presented on variants of the Grover's algorithm\cite{Cade_2023} and find an application in the present work.
However, we take the benchmarks beyond the ambiguous oracle count and compare our method to state-of-the-art classical solvers. 
QTG \cite{wilkening2024quantumalgorithmsolving01, wilkening2025quantumsearchmethodquadratic} is a search-based quantum algorithm with a problem-specific state preparation routine but is limited to purely positive constraints. We aim to generalize the ideas presented to more general problem structures in this work.
Quantum branch-and-bound \cite{Montanaro_2020} is another promising candidate for an application of amplitude amplification in combinatorial optimization and may offer a speedup over branch-and-bound methods as found in classical state-of-the-art solvers \cite{gurobi}. 
The related problem class of linear programs with continuous variables was tackled via a quantum version of the simplex algorithm in \cite{huang2021branchboundmixedinteger, nannicini2022fastquantumsubroutinessimplex} and analyzed with hybrid benchmarking techniques in \cite{ammann2023realisticruntimeanalysisquantum}. 
We emphasize that there is remarkable amount of work to solve combinatorial optimization problems with quantum heuristics such as the QAOA \cite{farhi2014quantumapproximateoptimizationalgorithm}.
However, we ignore these methods as they elude from a proper runtime analysis, especially for the large instances considered in this work.

\section{Constraint-oriented Biased quantum search}

\subsection{Single linear constraints}
In this section, the aim is to develop the state preparation technique, which lies at the core of the CBQS algorithm
Let $w \in \mathbb{Z}^n$ be a vector of $n$ integers and $C \in \mathbb{Z}$. We want to understand how to prepare a quantum state $\sum_{x: w^T x \leq C}a_x \ket{x}$.
\begin{definition}
We call a tuple $(w,C) \in \mathbb{Z}^n \times \mathbb{Z}$ a \textit{linear constraint}. A vector $x \in \mathbb{R}^n$ is said to satisfy the constraint if $w^T x \leq C$. Most vectors considered in this work will be bit strings $x \in \{0,1\}^n$. To each coefficient vector $w$ we associate a bit string $x^w$ which minimizes the constraint, i.e., $w^T x^w \leq w^T x$ for all $x \in \{0,1\}^n$. 
For each bit string, non-negative integer $i \leq n$ and $x \in \{0,1\}^n$ we define $P_i^w(x) \coloneqq \sum_{k=1}^i w_k x_k + \sum_{k=i+1}^n w_k x^w_k$.
\end{definition}
The bit string $x^w$ can be explicitly constructed:
\begin{equation}\label{eq:xw}
    x^w_i = 
    \begin{cases}
        0, &\text{if } w_i \geq 0\\
        1, &\text{else}.
    \end{cases}
\end{equation}
The justification for the state preparation method which is used in the related QTG algorithm \cite{wilkening2024quantumalgorithmsolving01} and will be adapted to more general constraints is the following fact.

\begin{lemma}\label{lemma:branching}
For any $x \in \{0,1\}^n$, we have $w^T x \leq C$ if and only if 
\begin{equation}\label{eq:branching_condition}
P_i^w(x) \leq C \text{ for all } i \in \{1,...,n\}    
\end{equation}

\end{lemma}
\begin{proof}
    The implication $\Leftarrow$ is clear as setting $i=n$ yields the claim.
    For the implication $\Rightarrow$ we expand and estimate the expression for $P_i^w(x)$ as follows.
    \begin{align*}
    \begin{aligned}
        &\sum_{k=1}^i w_k x_k + \sum_{k=i+1}^n w_k x^w_k =\\
        &w^T x + w^T x^w - w^T(x_1^w,...,x_i^w,x_{i+1},...,x_n) \\
        &\leq C + w^T x^w - w^T(x_1^w,...,x_i^w,x_{i+1},...,x_n)
        \leq C
    \end{aligned}
    \end{align*}
    The first ``$\leq$" is assumed, the second follows directly from the definition of $x^w$.
\end{proof}
This result allows us to set up an iterative method of constructing any bit string which satisfies the constraint $(w,C)$, namely by sequentially sampling bits such that \ref{eq:branching_condition} is satisfied at every step.

\begin{algorithm}
\caption{Sampling bit strings $y$ satisfying $w^T y\leq C$}
\label{alg:sampling}

\textbf{Require} Linear constraint $(w,C)$, sampling probabilities $\{p_0^i>0, p_1^i>0: i \in \{1,...,n\}\}$\;
\textbf{Ensure} A bit string $y \in \{0,1\}^n$ such that $w^T y \leq C$\;
    Initialize $x^w \in \{0,1\}^n$ according to \ref{eq:xw}\;
    $P \gets C - P_0^w(x) = C - w^{T} x^{w}$\;
\For{$i = 1$ to $n$}{
    \eIf{$P \geq |w_i|$}{
        Sample $y_i \in \{0,1\}$ with probabilities:\;
        $\mathbb{P}(y_i = x^w_i) = p_i$, $\mathbb{P}(y_i = 1 - x^w_i) = 1 - p_i$\;
        \If{$y_i = 1 - x^w_i$}{
            $P \gets P - |w_i|$\;
        }
    }{
        $y_i \gets x^w_i$\;
    }
}
\end{algorithm}

By Lemma \ref{lemma:branching}, algorithm \ref{alg:sampling} samples bit strings satisfying the constraint $(w,C)$. If additionally $0<p_0^i<1$ holds for the branching probabilities at every branching step $i$, the algorithm has a positive probability of  sampling any feasible bit string $x \in \{0,1\}^n$.

\subsection{Multiple linear constraints}

Now consider the situation where we have multiple constraints $(w^1, C^1),...,(w^m, C^m)$ and look for a bit string satisfying $(w^i)^T x \leq C^i$ for all $i$. A version  of lemma \ref{lemma:branching}, where we check condition \ref{eq:branching_condition} for every constraint no longer holds as is evident from the following simple example. Let $w^1 = (1,2)^T, C^1=2$ and $w^2 = (-1,-2)^T, C^2=-2$. Satisfying both constraints amounts to the condition $x_1 + 2x_2 = 2$. Observe that the bit string $x\coloneqq(1,0)^T$ clearly does not satisfy the constraint, but fulfills $P_1^{w^1}(x)\leq 2$ as well as $P_1^{w^2}(x)\leq -2$. It is of course fully expected that there exists no efficient (i.e. polynomial time) method to sample bit strings which satisfy an arbitrary list of linear constraints as this would among other results imply a polynomial solution of the NP-hard subset sum problem. 
However, we still want to make use of the sampling method \ref{alg:sampling} as the branching conditions do manage to exclude a significant number of infeasible bit strings in many practical cases.

\begin{algorithm}
\caption{Sampling bit strings trying to satisfy multiple constraints}
\label{alg:sampling_multiple_constraints}
\textbf{Require} Linear constraints $(w^k,C^k)$ for $k=1,...,m$, sampling probabilities $\{p_0^i>0, p_1^i>0: i \in \{1,...,n\}\}$\;
\textbf{Ensure} A bit string $y \in \{0,1\}^n$\;
    Initialize $x^{w^k} \in \{0,1\}^n$ according to \ref{eq:xw} for all $k=1,...,m$\;
    $P^k \gets C^k - P_0^{w^k}(x)$ for all $k=1,...,m$\;
\For{$i = 1$ to $n$}{
    $b_0 \gets True, b_1 \gets True$ \algorithmiccomment{Initialize branching conditions}\;
    \For{$k = 1$ to $m$}{
        \eIf{$w^k_i \geq 0$}{
            $b_1 \gets b_1 \wedge (P^k \geq w^k_i)$\;
            \algorithmiccomment{Update branching conditions}\;
        }{
            $b_0 \gets b_0 \wedge (P^k \geq -w^k_i)$\;
            \algorithmiccomment{Update branching conditions}\;
        }
        \eIf{$b_0 \text{ and } b_1$}{
            Sample $y_i \in \{0,1\}$ with probabilities:\;
            $\mathbb{P}(y_i = 0) = p^i_0$, $\mathbb{P}(y_i = 1) = 1 - p^i_1$\;
        }{
            \eIf{$b_{1-x_i^{w^1}}$}{
                $y_i \gets 1-x^{w^k}_1$\;
            }{
                $y_i \gets x^{w^k}_1$\;
            }
    }
    \eIf{$y_i = 1$}{
        $P_k \gets P_k - w^k_i$ for all $k$ such that $w^k_i > 0$\;
        \algorithmiccomment{Update constraint budget}
    }{
        $P_k \gets P_k + w^k_i$ for all $k$ such that $w^k_i < 0$\;
        \algorithmiccomment{Update constraint budget}
    }
    }
}
\textbf{return} $y = (y_1, \dots, y_n)$\;
\end{algorithm}

The working principle of algorithm \ref{alg:sampling_multiple_constraints} is the same as that of algorithm \ref{alg:sampling}. At each variable $x_i$ we check for all constraints whether assigning the value $x_i=0$ or $x_i=1$ would break the constraint in all cases, i.e. for all possible further assignments of the variables $x_{i+1},...,x_n$. Based on the result of the checks, one of four outcomes is possible. 
\begin{enumerate}
    \item Both assignments $x_0=0$ and $x_1=1$ are possible, i.e. it is still possible to satisfy every constraint individually. In this case, a variable assignment for $x_i$ is drawn randomly according to probabilities $\{p_0^i, p_1^i\}$ which are defined beforehand.
    \item Only assignment $x_0=0$ guarantees that every constraint can be satisfied individually. In this case we assign $x_0=0$.    
    \item Only assignment $x_0=1$ guarantees that every constraint can be satisfied individually. In this case, we assign $x_0=1$.
    \item Both variable assignments are guaranteed to break at least one of the constraints. In this case, we accept that the resulting bit string will be infeasible and (arbitrarily) assign $x_i=x_{x_i}^{w^i}$ and/or return ``infeasible".
\end{enumerate}

The sampling probabilities $\{p_0^i, p_1^i\}$, which come into play in the first case, can be chosen in different ways. Their importance is revealed when we turn to constrained minimization problems, where we aim to find solutions that are not only feasible but also have a small objective value. The most important principle used in this work and in \cite{wilkening2024quantumalgorithmsolving01} is at each step to preferably sample the variable assignment which can be found in a reference bit string $x^*$, usually the best currently known solution. Concretely, we then have
\begin{equation}
    p_{x^*_i}^i = \frac{b+1}{b+2}
\end{equation}
where $b>0$ is a tunable parameter. Intuitively, choosing a larger $b$ will give stronger preference to bit strings which are close in Hamming distance to the reference bit string $x^*$. A more detailed description of sampling strategies can be found in \ref{subsec:improvement_techniques}.

\subsection{From biased sampling to biased quantum search}
We now turn the classical methods outlined above into a quantum algorithm whose goal is to prepare a superposition
\begin{equation}
    \ket{\phi} = \sum_{y \in \{0,1\}^n} \sqrt{p_y} \ket{y},
\end{equation}
where $p_y$ denotes the probability that the bit string $y$ is sampled by the corresponding classical sampling algorithm. This can be done in quite a straightforward manner, namely by storing and constantly updating the constraint budget in a separate register and performing (biased) Hadamard gates on the main qubits controlled on the constraint register. More concretely, assume we have one problem qubit labeled $i$ holding the initial value for the $i$-th variable and one for the constraint budget labeled $c$ and define the unitary operation $C_i^w$ via 

\begin{equation}\label{eq:controlled_Hadamard}
    C_i^w: \ket{x_i}_{i}\ket{P}_c \mapsto 
    \begin{cases}
        \Tilde{H}\ket{x_i}_i\ket{P}_c, \text{ if } P \geq \vert w \vert \\
        \ket{x_i}_{i}\ket{P}_c, \text{ otherwise}
    \end{cases}
\end{equation}

\input{branching_circuit}

where $\Tilde{H}$ is a biased Hadamard gate and $P$ is the remaining constraint budget. It is described in \cite{wilkening2024quantumalgorithmsolving01}, Appendix A.4.e how to realize such a controlled operation with a series of multi-controlled gates.
After applying the controlled Hadamard, an update of the constraint budget register is performed via a subtraction operation controlled on the register of the variable, which is now in superposition. Depending on the sign of the constraint coefficient, the control value is $1$ or $0$.

\begin{equation}\label{eq:controlled_subtraction_1}
    S_i^1: \ket{x_i}_{i}\ket{P}_c \mapsto \ket{x_i}_{i}\ket{P - w_i x_i}_c
\end{equation}

\begin{equation}\label{eq:controlled_subtraction_0}
    S_i^0: \ket{x_i}_{i}\ket{P}_c \mapsto \ket{x_i}_{i}\ket{P - w_i (1-x_i)}_c
\end{equation}
A description of how to efficiently realize this type of operation based on QFT adders \cite{draper2004logarithmicdepthquantumcarrylookaheadadder} can be found in \cite{wilkening2024quantumalgorithmsolving01}.

In order to realize \ref{alg:sampling} as a quantum algorithm, we simply interleave the operations \eqref{eq:controlled_Hadamard}, \eqref{eq:controlled_subtraction_1} and \eqref{eq:controlled_subtraction_0} as follows:

\begin{equation}
U \coloneqq \prod_{i=1}^n C_i^{w_i} S_i^{(sign(w_i)+1)/2}
\end{equation}
Applying $U$ to the state $\ket{x^w}\ket{P_0^w}$ then yields a superposition state
 \begin{equation}
     \sum_{x \in \{0,1\}^n: x^T w \leq C}  \sqrt{p_x}\ket{x}\ket{C-w^T x}
 \end{equation}
 such that $p_x$ is exactly the probability with which algorithm \ref{alg:sampling} samples $x$.
Transferring algorithm \ref{alg:sampling_multiple_constraints} to a quantum circuit uses the same building blocks and is therefore analogous. 
In most applications of our state preparation technique, alongside the constraint(s) there is an objective function $f:\{0,1\}^n \rightarrow \mathbb{Z}$ to be minimized. It can be useful to track not only the constraint consumption but also the objective value for each bit string. This can be achieved in the same way as the constraint updates, namely with controlled additions/subtractions onto a separate objective register, see e.g. \cite{wilkening2024quantumalgorithmsolving01}. The quantum state $\ket{\psi} = U\ket{x^w}\ket{P_0^w}$ resulting from our state preparation unitary is then subjected to an amplitude amplification protocol in order to increase the measurement probabilities for states which exhibit an objective value lower than some threshold $T$.
The amplitude amplification operator is given by
\begin{equation}
    \mathcal{A} = U S_0 U^{\dagger} S_f
\end{equation}
with $S_0 = I - 2 \ket{0}\bra{0}$ being the reflection around the zero state and $S_f$ being the oracle operator with the action $S_f \ket{x} = -1^{g(x)} \ket{x}$, where $g=\mathbf{1}_{\{x:f(x)<T\}}$.
The complete amplitude amplification protocol \cref{alg:amplitude_amplification} is obtained by randomly choosing the number of AA iterations from a growing interval until either an improving solution is found or a termination criterion is reached. In our case, the termination criterion is defined as reaching a number of rounds chosen such that the expected total number of AA iterations matches the iteration limit $T>0$ (see \cref{fig:benchmarking_probs}).

\begin{algorithm}[t!]
\caption{Amplitude Amplification protocol for finding a bit string $x$ which improves the incumbent solution $x_0$ with respect to the objective function $f$.}
\label{alg:amplitude_amplification}
\textbf{Require} Objective function $f:\{0,1\}^n \rightarrow \mathbb{Z}$, feasible set $F \subset \{0,1\}^n$ such that membership $x \in F$ can be determined efficiently, incumbent solution $x_0 \in F$ with value $v \coloneqq f(x_0)$, iteration limit $T>0$, quantum state $\ket{\psi}$ to be ``searched"\;

$m \gets \lfloor \log_2(T) \rfloor$\;
Sample $r$ from $\{m, m+1\}$ with probabilities $p_m = 1-(\log_2(T)-n)$ and $p_{m+1}=1-p_m$\;
\For{$k$ in range($r$)}{
    Choose $j$ uniformly at random from $\{0,...,2^k - 1\}$\;
    Apply $\mathcal{A}^j$ to $\ket{\psi}$ and measure to obtain candidate solution $\ket{x}$\;
    \If{$x \in F$ and $f(x)<v$}{
        \textbf{return} $x$\;
    }
}
\textbf{return} $x_0$\;

\end{algorithm}

In this work, for benchmarking purposes, the problem we consider is the Knapsack problem with minimum filling constraint (MFKP) \cite{Xu2013Theknapsackproblemwithaminimumfillingconstraint}, that is written as:
\begin{align}
\begin{aligned}
	\max &\sum_{j = 1}^n x_j p_j\\
	s.t. 	&\sum_{j = 1}^n x_j w_j \leq c\\
		&\sum_{j = 1}^n x_j w_j \geq c - \epsilon\\
		& x_j \in \{0, 1\} &\forall 1\leq j \leq n.
\end{aligned}
\tag{MFKP}
\label{eq:gkp}
\end{align}
An excerpt of the respective state preparation quantum circuit is depicted in \cref{fig:circuit}.

\input{circuit}

\subsection{Improvement techniques}\label{subsec:improvement_techniques}

The provided algorithm is a general framework that can be improved using a range of techniques, from circuit optimization to classical machine learning.
Here we present a few of those potential improvements.

\subsubsection{Advanced weak biasing}
\begin{figure}[t]
	\centering
	\includegraphics[width=\linewidth]{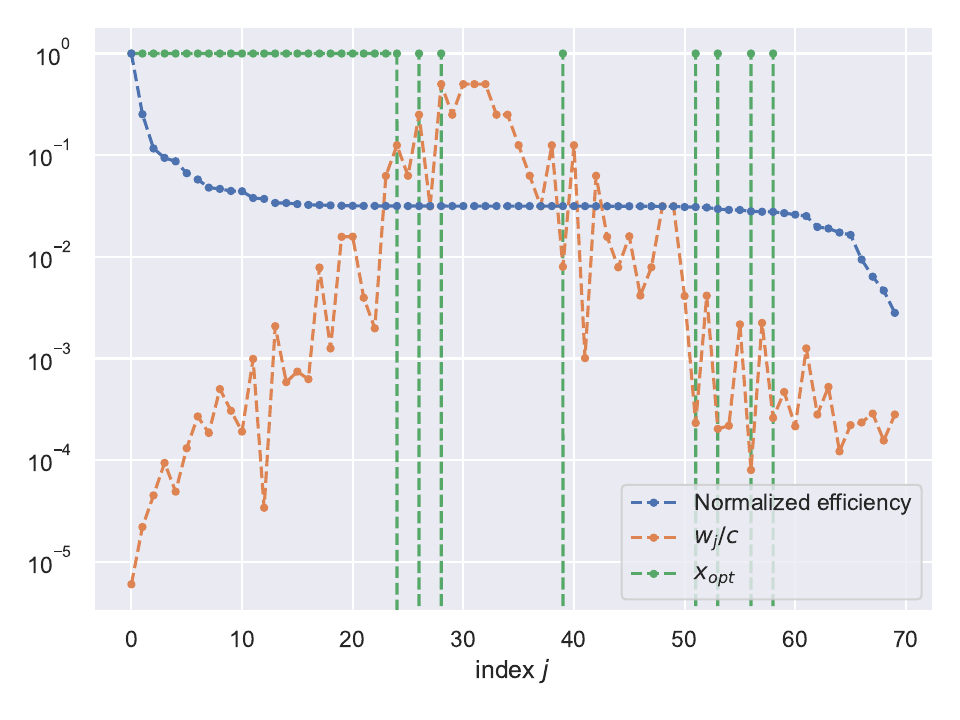}
	\caption{Properties of a 70 variable instance investigated in this work for the \ref{eq:gkp}.}
    \label{fig:opt-sol}
\end{figure}

Although the strategy employed in previous work \cite{wilkening2024quantumalgorithmsolving01, wilkening2025quantumsearchmethodquadratic} of biasing every assignment towards a known assignment has proven effective, this method can be further improved by incorporating more details from a given instance into the biasing strategy.
With the optimal biasing strategy, the optimal solution would be found quickly and with certainty; however, this requires too much information to be practical.
\cref{fig:opt-sol} shows some properties of a 70-variable instance as well as the assignment of the optimal solution $y$, indicating more information that can be incorporated into the biasing strategy.
Based on the figure, efficiency and relative cost compared to capacity might be good factors to incorporate into the biasing function.
Previously, the bias for assigning item $j$ was chosen, such that the probabilities of measuring assignments $0$ and $1$ are given by.
\begin{align}
\begin{aligned}
    &\left|\alpha_{0, j}\right|^2 = \frac{b (1-y_j) + 1}{b + 2}\text{ and}\\
    &\left|\alpha_{1, j}\right|^2 = \frac{b y_j + 1}{b + 2},
\end{aligned}
\end{align}
where $b = n / 4$. As this has proven very successful \cite{wilkening2024quantumalgorithmsolving01, wilkening2025quantumsearchmethodquadratic}, it will remain part of the biasing function and be extended by a different function $g: \{1,...,n\} \rightarrow [0,1]$.
The ratio between the extension function $g(j)$ and the previous biasing function, plays an important role and is therefore determined by an additional optimization parameter $f$:
\begin{align}
\label{eq:probs}
\begin{aligned}
    &\left|\alpha_{0, j}\right|^2 = 
    \frac{1}{1 + f}\left(\frac{b (1-y_j) + 1}{b + 2} + f(1-g(j))\right)
    \text{ and}\\
    &\left|\alpha_{1, j}\right|^2 = 
    \frac{1}{1 + f}\left(\frac{b y_j + 1}{b + 2} + fg(j)\right),
\end{aligned}
\end{align}
This method has the benefit of not increasing the complexity of the quantum circuits.

As an example for a refinement to the biasing function we define
a biasing function as a linear interpolation between the minimum and maximum cost relative to the constraint, if the item's efficiency $p_j / w_j$ is above a certain threshold:

\begin{align}
	\label{eq:biasing_strategy}
	g(j) = 
	\begin{cases}
		\frac{-0.6}{\max(r) - \min(r)} \left( \frac{w_j}{c} - \min(r) \right) + 0.8 &\text{if } \frac{p_j}{w_i} > 1\\
		0.2 &\text{else}.
	\end{cases}
\end{align}

Here $r = \{\frac{w_j}{c}: 1\leq j \leq n \}$.
The list of criteria which might lead to a refined biasing function presented in this paper is not exhaustive and highly problem-dependent.

\subsubsection{Look-ahead}
Another potential state preparation improvement is incorporating a subroutine called \textit{look-ahead}. 
This technique is commonly used within modern satisfiability solvers \cite{heule2009look}.
The basic premise of this technique is to more carefully choose the next variable assignment given a previous selection by an exhaustive search on a few variables $d \ll n$. 
It is therefore a technique built as a strong biasing method.
If, for example, a dead end for a specific variable assignment is discovered, this path can be avoided altogether. 
Within the quantum algorithm, this would result in the fact that the amplitudes of the feasible states are potentially boosted, reducing the number of Grover iterations to find these states.
Nevertheless, this exhaustive search comes at a cost. In the worst case, the runtime of the state preparation is scaled up by an exponential factor in the depth of the look-ahead.
Before we discuss an efficient implementation of the subroutine, we introduce some notation.

\begin{definition}
	A \text{quantum cycle} (or short \textit{cycle}) refers to the cost of a single quantum logic operation, defined in a quantum gate set, comprised of all the quantum gates that can be executed in parallel due to their action on disjoint sets of qubits.
\end{definition}

\begin{definition}
	Let $\mathcal{H} = \mathcal{H}^1\otimes \mathcal{H}^2 = \mathbb{C}^k\otimes \mathbb{C}^k$ be a Hilbert space comprising two registers of the same size. 
	Each register can be realized by $\log_2 k$ qubits.
	We define the copy of the binary representation of an integer stored in $\mathcal{H}^1$ into $\mathcal{H}^2$ by the unitary operation
	\begin{align}
		U^{1,2}_{copy} = \prod_{j = 1}^{\log_2 k} C^{1,j}(X^{2,j}),
	\end{align}
	where controlled by qubit $j$ in register 1, qubits $j$ in register 2 will be flipped.
\end{definition}

Since none of the controlled gates making up $U^{1,2}_{copy}$ act on the same qubits, the copy unitary can be executed within a single cycle.

\begin{lemma}
	Assume we are given $(m + 1)$ $k$-quantum registers $\mathcal{H}^j = \mathbb{C}^k$, where register one stores an integer $j$ and the remaining registers store the value 0.
	Creating m copies of the integer $j$, there exists a quantum algorithm requiring $\mathcal{O}(\log_2 m)$ cycles.
\end{lemma}
\begin{proof}
	Given the Hilbert space $\mathcal{H}=\mathcal{H}^1 \otimes \cdots \otimes \mathcal{H}^{m+1}$, we can define the unitary
	\begin{align}
		U^j = \prod_{i = 1}^{2^{j-1}} U^{i, i + 2^{j-1} + 1}_{copy}
	\end{align}
	The copy operations within $U^j$ act on disjoint registers, which is why they can be implemented within a single cycle.
	Consequently, applying $U^1$ up to $U^{\log_2 m}$ doubles the number of copies and therefore requires only $\log_2 m$ cycles.
\end{proof}

Now we can discuss the efficient sub-exponential runtime implementation of the look-ahead quantum subroutine.
\begin{theorem}
\label{theorem:look-ahead-depth}
	Let $d>0$ and $c>0$ be integers. There exists a quantum algorithm that can implement the look ahead for the \ref{eq:gkp} with  $\mathcal{O}(nd)$ additional cycles utilizing $\mathcal{O}(2^{d+1}\log_2 c)$ additional qubits compared to the plain implementation.
\end{theorem}

\begin{proof}
\label{theorem:look-ahead}
	To implement a look ahead of depth $d$ when considering assignment $x_j$ requires the investigation of all the $2^d$ potential assignments $F=\{k:k\in\mathbb{N}, 0 \leq k\leq 2^d - 1\}$.
	Instead of doing the constraint checks for all the weights individually for every assignment $k\in F$, we can precompute the two sets
	\begin{align}
    \begin{aligned}
		P_1 =& \left\{\sum_{i = 0}^d  b_i(k) w_{i+j}: k \in F \right\} \text{ and }\\
		P_2 =& \left\{\sum_{i = 0}^d  \left( 1 - b_i(k) \right)w_{i+j}: k \in F \right\}, \\&\text{where } b_i(k) \coloneqq \left\lfloor \frac{k}{2^i} \right\rfloor\mod 2.
    \end{aligned}
	\end{align}
	
	Given the Hilbert space $\mathcal{H} = \mathcal{H}^1 \otimes \mathcal{H}^2 \otimes \cdots \otimes \mathcal{H}^{2^{d+1}}  = \mathbb{C}^{\log_2 c} \otimes  \cdots \otimes \mathbb{C}^{\log_2 c}$, where register 1 and 2 store the remaining capacities of the two constraints. 
	Applying $\prod_{j = 1}^{d} U^j$ on the subspaces $\mathcal{H}^1 \otimes \mathcal{H}^3 \otimes \cdots \otimes \mathcal{H}^{2^d}$
	and $\mathcal{H}^2 \otimes \mathcal{H}^{2^d + 1} \otimes \cdots \otimes \mathcal{H}^{2^{d+1}}$ creates the necessary number of copies within $d$ cycles. 
	Afterwards, for every value $k\in P_1, P_2$, the algorithm checks if $k$ respects the constraints, which will be stored in ancilla qubits.
	This can be done fully in parallel (as well as the uncomputation), not adding any time overhead to the algorithm.
	Instead of just checking one boolean value that determines whether the assignments for $x_j$ are possible, we have to check whether we can find at least one satisfying assignment among all $2^d$ possible assignments that originate in $x_j = 0$ or $1$.
	This operation can be performed by a single multi-controlled gate on all the booleans, resulting in an additional time overhead of $2\log_2 2^d = 2d$ cycles.
	Finally, uncomputing all the additional registers adds $d$ cycles.
	Compared to the plain implementation, applying the look ahead on every variable $x_j$ for $1\leq j \leq n$ adds a time overhead of $4nd$.
\end{proof}

The look-ahead implementation comprises $\mathcal{O}(n2^d)$ quantum gates, which require error correction.
This will result in a quasi-polynomial logarithmic time overhead \cite{error_correction}, yielding an $\mathcal{O}(poly(d))$ overhead. Consequently, this routine can still be efficiently implemented, provided error-corrected qubits and a low look-ahead depth $d \ll n$.

\subsubsection{Look-ahead biasing}
More information gathered by the look-ahead could be incorporated into the state preparation.
As described in the previous section, the look-ahead is used as a strong biasing strategy.
Instead of just checking whether any of the assignments $x_j$ and $\lnot x_j$ may lead to a non-clause-violating assignment, we could track on both branches, how many potential candidate solutions remain, and apply a respective rotation gate accordingly.
For example, if a more uniform distribution is desired, the amplitudes could be adjusted accordingly so that all potential candidate solutions are weighted similarly.
It is also important that this extension is efficiently implementable.
\begin{lemma}
    Let $d > 0$ be an integer.
    There exists a quantum algorithm that implements the look-ahead biasing using $\mathcal{O}(d\log_2 d)$ cycles and $\mathcal{O}(2^{d+1}d)$ additional qubits.
\end{lemma}
\begin{proof}
    To count the potential candidates for both assignments, $2^d$ boolean variables have to be summed up. 
    This can be done efficiently utilizing the \textit{divide-and-conquer} technique \cite{SMITH198537}, where pairwise additions are done in parallel.
    As the maximum value cannot surpass $2^d$, we require $d$ qubits to store the result, which leads to additions with depth $\mathcal{O}(\log_2 d)$ \cite{draper2004logarithmicdepthquantumcarrylookaheadadder}.
    Based on the counts, controlled rotations are applied to the item assignment.
    If we allow at most $d$ case distinctions, which should be sufficient, we achieve a runtime-efficient implementation.
\end{proof}

For the considered problem, we didn't achieve any improvements with this technique, which is why it is not included in the results. 
This extension may have a greater impact on different problems classes or with more advanced utilization.

\subsubsection{Item ordering}
Another important consideration is the order in which items are evaluated. 
This does not require a change to the prescribed initial algorithm, yet it might significantly impact its performance.
Even though the \ref{eq:gkp} has a straightforward property to sort the items, the efficiency $p_j / w_j$, this might not apply to more general and complex problems.
We investigate the effects of various item orderings based on the properties of the considered problem, as well as a random ordering, to provide more insight into the method's performance for problems without a clear ordering prescription.

\begin{figure*}[t]
    \centering
    \includegraphics[width=\linewidth]{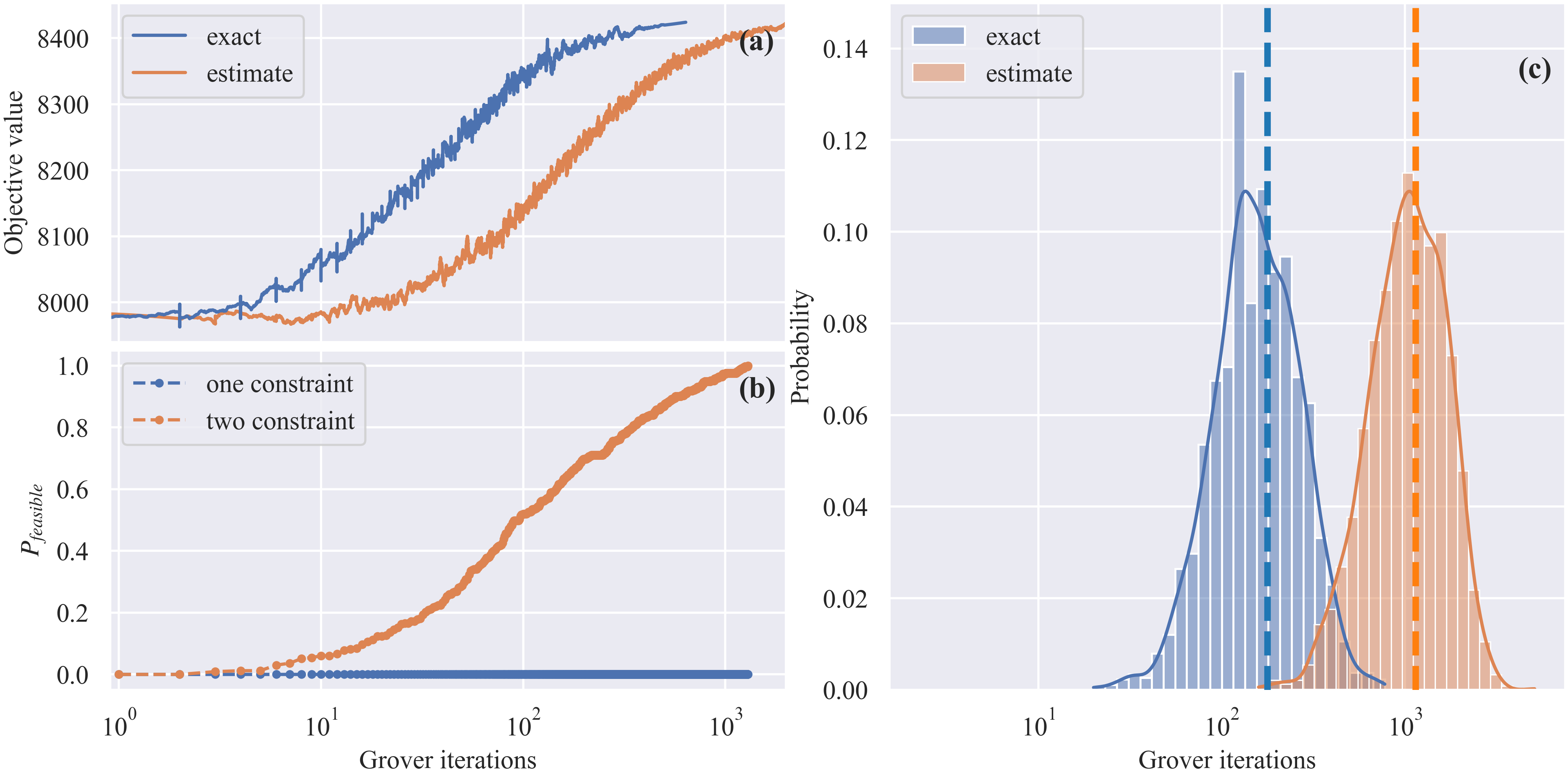}
    \caption{
    \textbf{(a)} Behavior of the objective value over time of the exact against the sampling-based benchmarking method.
    \textbf{(b)} Probability of finding a feasible solution when incorporating both or just one of the two constraints in the state preparation, tested on a non-trivial 70-variable instance with random ordering of items.
    \textbf{(c)} Number of oracle calls of the exact and sampling-based heuristic benchmarking method to find the optimal solution of a 25-variable instance.
    We observe that approximate benchmarking provides an empirical lower bound on the number of oracles and that significantly fewer Grover iterations are required to find a feasible state when incorporating both constraints.
    }
    \label{fig:general_properties}
\end{figure*}

\section{Benchmarking}
There are two main ways to benchmark our algorithm. One is to use a hybrid benchmarking scheme \cite{Cade_2023}, i.e., classically finding the list of states exceeding a certain objective threshold and calculating an amplitude for each of them. Then one can analytically compute how the amplitudes change after a certain number of AA iterations. 
A more detailed description can be found in the supplementary information of \cite{wilkening2024quantumalgorithmsolving01}.
The bottleneck in this approach is surely the compilation of all states exceeding the initial threshold, and we find that it already becomes intractable for instances with fewer than $100$ variables.
For poorly chosen thresholds, which might not be avoidable in complex problems and instances, even $30-40$ variables can exceed the capabilities of computing systems.
To circumvent this issue, we instead resort to an approximate benchmarking scheme. More precisely, we sample from algorithm \ref{alg:sampling_multiple_constraints} a number of times $\Tilde{T}>0$ and find a corresponding ``iteration maximum" $T$ for which algorithm \ref{alg:amplitude_amplification} would offer almost the same chances of success as its classical counterpart. With the approximate iteration limit $T$, we benchmark quantum resources by counting the gates required to perform $T$ Grover iterations. In the benchmarking for the present work, we chose $\Tilde{T}=T^2$ which corresponds to the intuition that Amplitude Amplification should offer a quadratic speedup compared to repeated sampling of the initial state.
Let $p \in [0,1]$ be the probability that \cref{alg:sampling_multiple_constraints} samples a good state. Consequently the amplitude of these states after state preparation will be $\sqrt{p}$. Now, sampling according to the initial probabilities $T$ times yields a success probability of 
\begin{equation}
    P_T = 1 - (1-p)^T.
\end{equation}
On the other hand, applying the Amplitude Amplification protocol \ref{alg:amplitude_amplification} with an iteration maximum of $T$ yields the success probability

\begin{align}
    P_T^{AA} = &(1-w)\prod_{r=1}^m \sum_{j=0}^{2^r - 1} \frac{1}{2^r} \cos^2(2(j+1)\theta)+\\
    &w\prod_{r=1}^{m+1} \sum_{j=0}^{2^r - 1} \frac{1}{2^r} \cos^2(2(j+1)\theta)
\end{align}

where $m \coloneqq \lfloor \log_2(T) \rfloor$, $w \coloneqq \log_2(T)-m \in [0,1]$, and $\theta \coloneqq \arcsin(\sqrt{p})$.
The expressions are hard to compare analytically, though a concrete example of the success probabilities with $p=0.001$ is given in \cref{fig:benchmarking_probs}. We observe that especially in the very relevant regime of slightly lower success probabilities, the curves behave very similarly and if anything our benchmarking scheme paints a slightly pessimistic picture.
It cannot be avoided that in the regime of a high iteration maximum, the classical protocol exhibits a higher success probability as, in contrast to the AA protocol, it does not have the chance of ``overcooking" the number of iterations.
\begin{figure}[ht]
    \centering
    \includegraphics[width=0.8\linewidth]{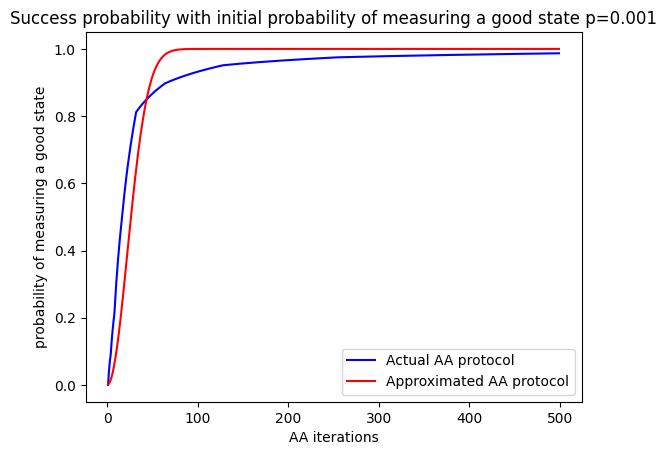}
    \includegraphics[width=0.8\linewidth]{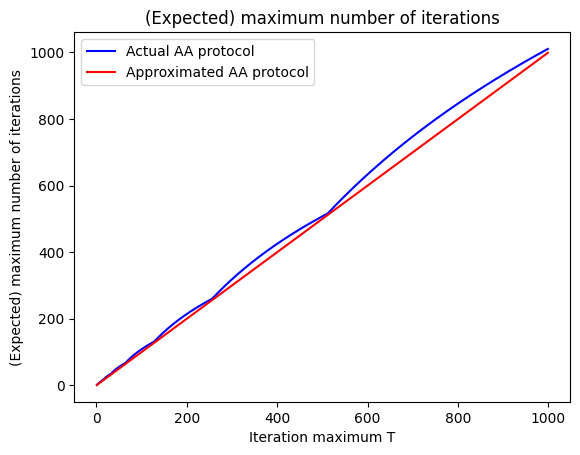}
    \caption{Comparison of our benchmarking strategy (red) and the actual amplitude amplification protocol suggested in this work (blue). We compare both success probability and expected number of iterations and find that the benchmarking protocol behaves very similarly in both metrics and can thus be used to produce meaningful benchmarks for the suggested quantum algorithm.}
    \label{fig:benchmarking_probs}
\end{figure}

\subsection{Comparisons}
When comparing algorithms' efficiency, not only the time to find the optimal solution or prove its optimality is of interest, especially for very complex problems, but also finding good incumbent solutions is important.
An incumbent solution is a newly found solution with an improved objective value compared to the current best-found assignment.
For this reason, we benchmark all the algorithms that are our quantum algorithm (CBQS), Gurobi, simulated annealing, and a generalized Goemans-Williamson, by storing the time stamp and objective value of various incumbent solutions.
To estimate the runtime of our quantum algorithm, we have to make some assumptions about the quantum hardware's capabilities. 
As we don't aim to provide an actual quantum advantage assessment, rather we want to show the potential for a far-term quantum advantage, we assume an optimistic, but realistic quantum hardware.
We assume that 
\begin{itemize}
    \item all consecutive gates acting on disjoint sets  of qubits can be performed in parallel, 
    \item (single controlled) arbitrary single qubit unitaries as well as the Toffoli, can be executed with the cost of a single gate
    \item logical gates can be performed in a time of $10^{-8}s$.
\end{itemize}
While the logical cycle time can be seen as \emph{very optimistic}, any overhead due to factors such as error correction can be absorbed into the time, thereby adjusting the results.
Yet this is beyond the scope of this paper.

All experiments are performed on an Intel Core i7 with $6$ cores ($12$ threads), a 2.6 GHz (4.1 GHz) clock speed, and 16 GB RAM.
The source code, experiments, and data are available at \cite{repo-code, repo-data}.

\subsection{Goemans-Williamson inspired algorithm for the 0-1 Knapsack problem}
There is a way to approximate the optimization value of the problem \eqref{eq:gkp} by sampling from a distribution obtained from a semi-definite relaxation of the problem.
In \cite{Lasserre2016_MaxCutFormulation}, the author presents an approach to encode any 0-1 quadratic problem with linear constraints in a MAXCUT problem, which shall be explained more explicitly below.
Apart from achieving good outer bounds to use in Branch-and-Bound-like schemes, the optimal solution of the semi-definite program can be used to sample bit strings.
For the original MAXCUT problem, this method has been proven to be optimal, delivering solutions of mean value at least 0.87856 times the optimal value \cite{GoemansWilliamson}.
An empiric study of this versatile approach on more  optimization problems is conducted in \cite{GoemansWilliamsonForBQP}

In a first step, we promote the inequality constraints to equality constraints by introducing an auxiliary, so called slack variable \( 0 \le y \le \varepsilon \) and its binary representation \( (x_{n+1}^{}, \dots, x_{n+s+1}^{}) \), i.e., \( y = \sum_{j=1}^{s} 2^{j} x_{n+1+j}, \ s = \lceil \log \varepsilon \rceil \).
Next, any binary function \( f: \{0, 1\} \to \mathbb{R} \) can be transformed to a function on the nodes of the \( n \)-dimensional hypercube \( f_{}^{\prime}: \{-1, 1\} \to \mathbb{R} \) by the variable transformation \(x_{j}^{} \mapsto 2 x_{j}^{} - 1 \).
Therefore, \eqref{eq:gkp} is equivalent to an optimization over the \( n \)-dimensional hypercube with \( p \mapsto p/2, \ w \mapsto w/2 \), and \( c \mapsto c - w_{}^{T} \bm{1}\) where \( \bm{1} \) is the vector of all ones.
The quadratic form that we obtain by appending the constraints to the objective function as a penalty and homogenization,
\begin{equation}
    Q(x_{0}^{}, x) = x_{0}^{} p_{}^{T} x + (2 \cdot p_{}^{T} x + 1) \cdot \left\| w_{}^{T} x - x_{0} c\right\|
\end{equation}
is equivalent to \eqref{eq:gkp} \cite[Theorem 2.2]{Lasserre2016_MaxCutFormulation}.
Solving the simple semidefinite program \( \max \{ \langle \bm{Q}, \bm{X} \rangle \ \vert \ X_{ii} = 1 \ \forall i = 1, \dots, n+s+2, \ \bm{X} \succeq 0 \} \) where \( Q(x_{0}^{}, x) = (x_{0}^{}, x) \bm{Q} (x_{0}^{}, x)^{T} \) thus yields an upper bound to the original problem.

Due to its semi-definiteness, the optimal solution is a Gram matrix whose vectors can be recovered by, e.g., an eigenvalue decomposition of \( \bm{X} \).
The hyperplane rounding and subsequent post-selection on constraint satisfaction of course does not recover the approximation bound of the original Goemans-Williamson algorithm, but performed in parallel on multiple cores yields some non-negligible success probability.

\begin{figure*}[t!]
    \centering
    \includegraphics[width=\linewidth]{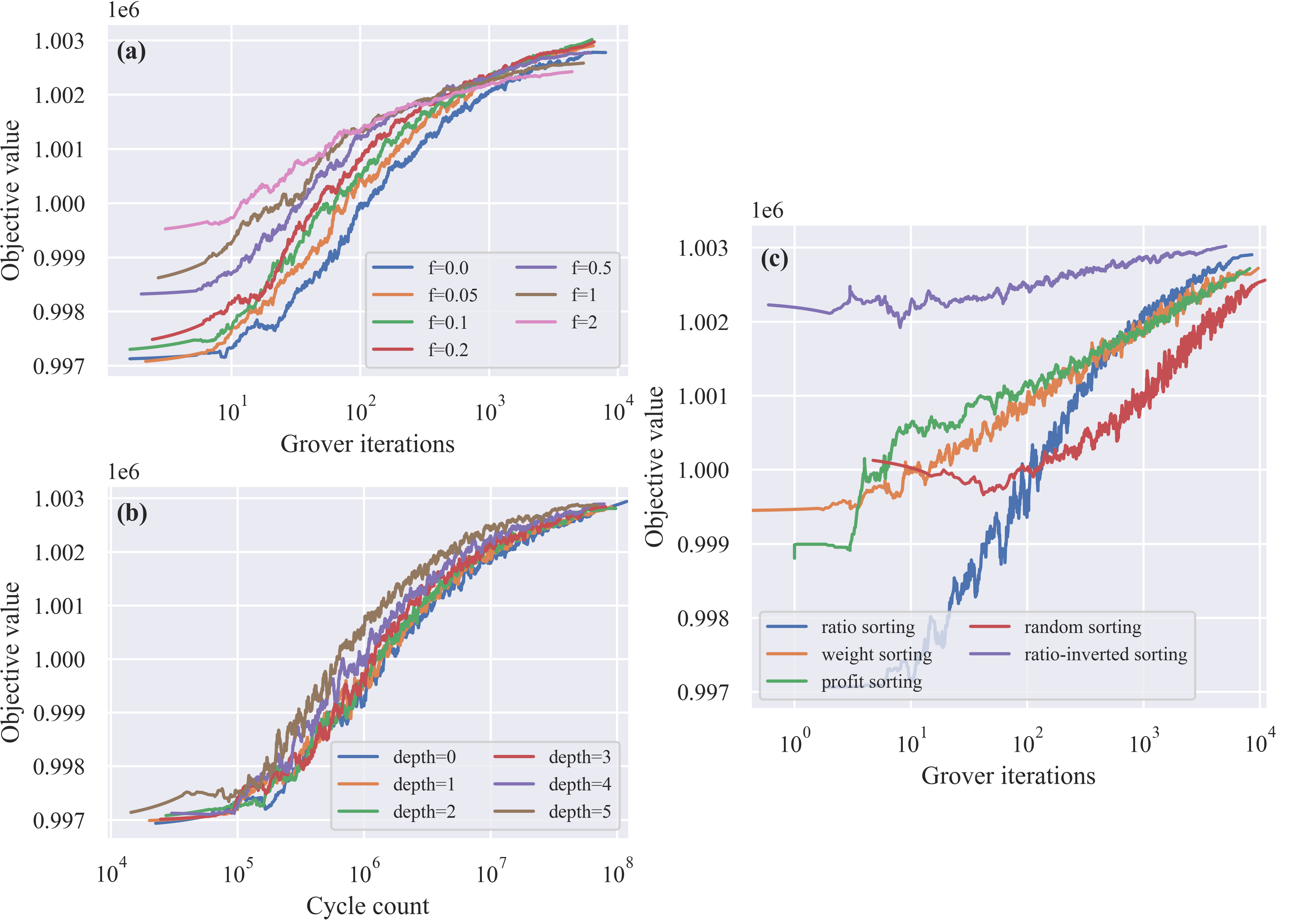}
    \caption{
        \textbf{(a)} Influence of the biasing factor $f$ mixing the original biasing strategy with a manual function, further considering the structure of a given instance of the \ref{eq:gkp}.
        Given an appropriate biasing function, the mixing factor can significantly improve runtime performance, especially in the early stages of the search.
        \textbf{(b)} Performance of the biased quantum search when introducing the look-ahead subroutine in the state-preparation circuit. 
	    We observe that a deeper look-ahead reduces the number of Grover iterations required to find the incumbent solutions.
        \textbf{(c)} The number of Grover iterations required to find incumbent solutions based on the variable sorting of a non-trivial 70-variable instance of the \ref{eq:gkp}.
    }
    \label{fig:improvemts}
\end{figure*}

\begin{figure*}[t]
    \centering
    \includegraphics[width=\linewidth]{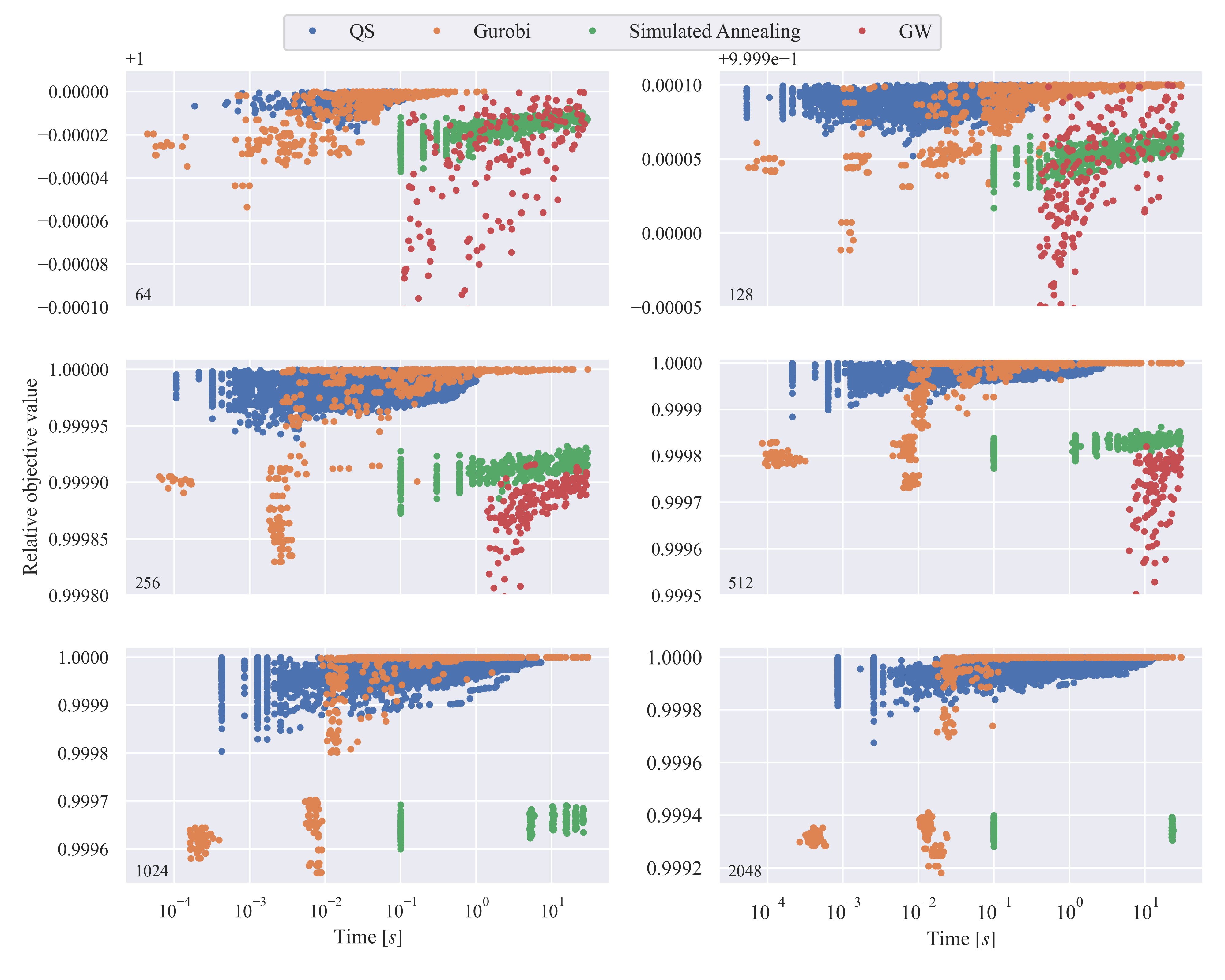}
    \caption{Comparison of our biased quantum search method with the classical methods Gurobi, Simulated Annealing, and Goemanns Williamson.
    We observe that our algorithm has the potential to outperform classical methods in finding incumbent solutions.
    Gurobi usually performs better at finding near-optimal and optimal solutions, but our method can often find better solutions early.}
    \label{fig:main_comparison}
\end{figure*}

\section{Results}

In this section, we provide the numerical results of our quantum algorithm.
Firstly, we discuss some properties of the benchmarking techniques and improvements specific to the CBQS presented in the previous section.
Afterward, we compare the algorithm's performance with classical algorithms, especially the state-of-the-art solver Gurobi.

\subsection{General Results}
For large, complex instances, our exact benchmarking method fails due to runtime and memory constraints; we aim to provide good runtime predictions for our quantum algorithm using a classical equivalent sampling strategy.
\cref{fig:general_properties} ((a) and (c)) shows that the exact and the sampling-based benchmarking routines predict similar behaviors of the objective value with increasing number of oracle applications.
Moreover, the exact benchmarking also predicts better performance for the quantum algorithm, suggesting that sampling-based benchmarking provides a lower bound on overall performance.
Due to the results, we will use the sampling-based benchmarking method for all the following experiments.

\cref{fig:general_properties} also provides results on the performance of CBQS when only a single or both constraints are incorporated into the state preparation (b).
Here, this is demonstrated on a 70-variable instance with randomly sorted items.
As expected, the probability of finding a satisfying assignment after a certain number of oracle applications drastically increases when considering both constraints.
When using only a single constraint, the algorithm fails to find any solution, even with up to 2000 oracle applications.

\subsection{Improvement techniques}

For \cref{fig:improvemts}, the properties of an advanced biasing function, the look-ahead with various depths, and the item ordering are benchmarked.
Firstly, we observe that an advanced biasing function taking the item's weight and profits into account can boost the algorithm's performance (\cref{fig:improvemts} (a)).
The advanced biasing function is weighted against the standard assignment-based bias with some factor $f$.
When increasing $f$, we observe that the objective value improves, especially in the early stages of the quantum search.
Yet a poorly chosen factor $f$ that is too large increases the required time to find the optimal solution.
Following a non-constant-factor implementation might be superior.

Another improvement can be made by incorporating the look-ahead with a depth of up to 5.
Here, it doesn't suffice to compare the number of oracle applications, as the implementation cost of the oracle increases with the look-ahead depth.
For the data in \cref{fig:improvemts} (b), the oracle's cycle count is computed in \cref{theorem:look-ahead-depth}.
Despite the increased oracle cost, the reduction in the number of oracle calls overcomes this, making the quantum algorithm more performant with increased look-ahead depth.
Even though a larger depth requires more qubits, which would enable more parallel repetitions of low-depth implementation, the additional relative qubit cost drops to 0 for $n\rightarrow \infty$.
Therefore, for $n\gg d$, no implementation allows for parallel execution given the same amount of available qubits.

Lastly, we investigate the algorithm's behavior with different item orderings.
As the instances are characterized by only a list of profits and a list of weights, the ordering criteria are the ratio of profit to weight, the inverse of that, ascending weight, and descending profit.
Additionally, we investigate a random item ordering, as for a more general problem, these ordering styles might not be available.
\cref{fig:improvemts} (c) shows that all orderings have similar performances, where the inverse ratio ordering emits the best overall computational efficiency.
Despite expectations that ratio sorting would perform best, it cannot compete with inverse ratio sorting. 
Here, the algorithm considers items in decreasing order of their item efficiency $p_j / w_j$.
This might be because the adjusted constraint 2, with its negative coefficients, will have a larger right-hand side, reflecting the constraints' potential.
By first assigning the inefficient items to $0$, the constraints will be tightened, leading to stronger biasing and increasing the probability of finding good solutions.
Surprisingly, the algorithm with randomly sorted items performed well, especially during the early stages of execution. 
Although this approach fails to return solutions of comparable quality to better-structured approaches, it still suggests that the quantum algorithm may perform well on problems with less obvious item ordering.

\subsection{Comparison with classical methods}

After discussing benchmark data focused solely on the quantum algorithm, we use the best implementation to compare it with classical algorithms.
For \cref{fig:main_comparison}, we investigate instances with an exponentially increasing number of variables from 64 to 2048.
We observe that the considered heuristics, simulated annealing and Goemanns Williamson, are strictly outperformed by Gurobi and CBQS.
Additionally, in the early stage of the search, CBQS, which provides the mentioned hardware capabilities, can offer better incumbent solutions earlier than Gurobi.
Even though Gurobi eventually surpasses the performance of quantum algorithms, we observe the potential for a quantum advantage.
This potential runtime savings likely increase with the practical difficulty of the problem/instances.

\section{Discussion}
Since we can demonstrate that our algorithm has the potential to provide a quantum advantage, given a proper quantum computer and fast logical gates, a more rigorous investigation of the error-correction overhead is required.
Our goal was to extend a known quantum algorithm to a broader class of optimization problems. 
Investigating the algorithms' capabilities for tackling different, more complex problems beyond linear ones could demonstrate better performance of quantum search compared to classical algorithms.
Similar techniques can be applied not only to constrained optimization problems but also to constraint satisfaction problems. For different problem classes, other biasing strategies might become attractive. Possible candidates for biasing criteria are the objective value of a variable, the constraint consumption of a variable, or a combination of both. A systematic exploration of biasing schemes could be an interesting topic for further research, including improvements using machine learning methods such as gradient descent \cite{ruder2017overviewgradientdescentoptimization}.
Amplitude Amplification-supported versions of more specific classical heuristics could be another fruitful research direction, the main question is whether one can find efficient circuits preparing ``superposition versions" of these heuristics.

\section{Acknowledgment}

We thank Tobias J. Osborne, René Schwonnek, and Lennart Binkowski for insightful discussions.
This project was enabled by the DFG through SFB 1227(DQ-mat), QuantumFrontiers, the QuantumValley Lower Saxony, the BMBF projects ATIQ and QuBRA, the BMWK project ProvideQ, and the Quantera project ResourceQ.

\bibliographystyle{unsrt}
\bibliography{sample}

\end{document}

%% file: branching_circuit.tex
\begin{figure}[t]
\begin{quantikz}[column sep=0.2cm, row sep=0.3cm]
  \lstick{$\ket{x_i}$} &\qw & [1.5cm] & \gate{\Tilde{H}} & \ctrl{1} & \qw \\
  \lstick{$\ket{P_i(x)}$} & \qwbundle{L} & [1.5cm] & \ctrl[wire style = {"P_i(x) \geq w_i"}]{-1} & \gate{-w_i} & \qw \\
\end{quantikz}
\caption{Circuit for conditional branching and constraint updating on a single variable. }
\label{fig:branching_circuit}
\end{figure}

%% file: circuit.tex
\begin{figure*}[t]
\resizebox{\textwidth}{!}{
\begin{quantikz}[column sep=0.2cm, row sep=0.3cm]
  \lstick[wires=3]{$\ket{x}$} & 
    \qw & 
    \qw & 
    \qw & 
    \qw & 
    \gate{R_y(\theta_1)} &
    \gate{X} & 
    \qw &
    \qw &
    \qw & 
    \octrl{5} &
    \ctrl{7} &
    \qw &
    \qw &
    \qw &
    \qw &
    \qw &
    \qw &
    \qw &
    \qw &
    \qw &
    \qw &\\
  & 
    \qw & 
    \qw & 
    \qw & 
    \qw & 
    \qw &
    \qw & 
    \qw & 
    \qw & 
    \qw &
    \qw & 
    \qw & 
    \qw & 
    \qw &
    \qw & 
    \gate{R_y(\theta_2)} & 
    \gate{X} & 
    \qw & 
    \qw & 
    \octrl{4} &
    \ctrl{6} &
    \qw &\\
  & 
    \qw & 
    \qw & 
    \qw & 
    \qw &
    \qw & 
    \qw &
    \qw &
    \qw &
    \qw &
    \qw &
    \qw &
    \qw &
    \qw &
    \qw &
    \qw &
    \qw &
    \qw &
    \qw &
    \qw &
    \qw &
    \qw &\\
  \lstick[wires=2]{$\ket{0}^2$} & 
    \qw & 
    \qw & 
    \targ{} & 
    \qw & 
    \ctrl{-3} &
    \octrl{-3} & 
    \targ{} & 
    \qw & 
    \qw & 
    \qw & 
    \qw & 
    \qw & 
    \targ{} &
    \qw & 
    \ctrl{-2} &
    \octrl{-2} &
    \targ{} &
    \qw &
    \qw &
    \qw &
    \qw &\\
  & 
    \qw & 
    \qw & 
    \qw & 
    \targ{} &
    \ctrl{-1} &
    \ctrl{-1} &
    \qw & 
    \targ{} &
    \qw & 
    \qw & 
    \qw & 
    \qw & 
    \qw & 
    \targ{} &
    \ctrl{-1} &
    \ctrl{-1} &
    \qw &
    \targ{} &
    \qw &
    \qw &
    \qw &\\
  \lstick[wires=2]{$\ket{P_1 = c-P_0^{w}}$} & 
    \qw & 
    \gate[2]{- w_1} & 
    \octrl{-2} &
    \qw &
    \qw &
    \qw & 
    \octrl{-2} &
    \qw &
    \qw &
    \gate[2]{+ w_1} &
    \qw &
    \gate[2]{- w_2} &
    \octrl{-2} & 
    \qw & 
    \qw &
    \qw &
    \octrl{-2} &  
    \qw &
    \gate[2]{+ w_2} &
    \qw &
    \qw &\\
  & 
    \qw & 
    \qw & 
    \qw &
    \qw &
    \qw &
    \qw &
    \qw & 
    \qw &
    \qw &
    \qw &
    \qw &
    \qw &
    \qw &
    \qw &
    \qw &
    \qw &
    \qw &
    \qw &
    \qw &
    \qw &
    \qw &\\
  \lstick[wires=2]{$\ket{P_2 = -c+\varepsilon-P_0^{-w}}$} & 
    \qw & 
    \gate[2]{- w_1} & 
    \qw & 
    \octrl{-3} &
    \qw &
    \qw &
    \qw & 
    \octrl{-3} &
    \qw &
    \qw &
    \gate[2]{+ w_1} & 
    \gate[2]{- w_2} &
    \qw &
    \octrl{-3} & 
    \qw &
    \qw &
    \qw &
    \octrl{-3} & 
    \qw &
    \gate[2]{+ w_2} & 
    \qw & \\
  & 
    \qw &
    \qw & 
    \qw & 
    \qw &
    \qw &
    \qw &
    \qw &
    \qw &
    \qw &
    \qw &
    \qw &
    \qw & 
    \qw &
    \qw &
    \qw &
    \qw &
    \qw &
    \qw &
    \qw &
    \qw &
    \qw &\\
\end{quantikz}
}
\caption{State preparation quantum circuit example for a two-variable gapped knapsack problem with constraints $w^T x \leq c$ and $-w^T x \leq -(c-\varepsilon)$.
Rather than checking whether $w_j$ fits the remaining constraint, we can subtract its absolute value and check if the result is not negative by looking at the uppermost sign qubit in each of the lower two registers. 
Then we only have to do a controlled uncomputation after assigning the variable.}
\label{fig:circuit}
\end{figure*}